\documentclass{llncs}

\usepackage[utf8]{inputenc}
\usepackage[T1]{fontenc}
\usepackage{amsmath}
\usepackage{amssymb}

\usepackage[english]{babel}
\usepackage[pdftex,
            pdfauthor={Sebastian Buchwald, Manuel Mohr, Ignaz Rutter},
            pdftitle={Optimal Shuffle Code with Permutation Instructions},
            pdfsubject={Optimal Shuffle Code with Permutation Instructions},
            pdfkeywords={compiler, register, allocation, shuffle, code, permutation, optimal, greedy, dynamic, programming},
            pdfproducer={Latex with hyperref},
            pdfcreator={pdflatex}]{hyperref}
\usepackage{listings}
\usepackage{tikz}
\usetikzlibrary{positioning}
\usepackage[nomessages]{fp}
\usepackage{cleveref}
\crefname{figure}{fig.}{figures}
\usepackage{graphicx}
\usepackage{xspace}
\usepackage{todonotes}
\usepackage{subcaption}
\usepackage{paralist}
\usepackage{csquotes}
\usepackage{listings}
\usepackage{cite}
\usepackage{array}
\usepackage[norefs,nocites]{refcheck}
\usepackage[labelfont=bf]{caption}

\usepackage{times}
\makeatletter
\newcommand{\refcheckize}[1]{%
  \expandafter\let\csname @@\string#1\endcsname#1%
  \expandafter\DeclareRobustCommand\csname relax\string#1\endcsname[1]{%
    \csname @@\string#1\endcsname{##1}\wrtusdrf{##1}}%
  \expandafter\let\expandafter#1\csname relax\string#1\endcsname
}
\makeatother

\refcheckize{\cref}
\refcheckize{\Cref}

\newcommand{\cycle}[8]{%
	\foreach \s in {1,...,#1}%
	{%
		\FPeval{\nodeidx}{clip(\s+#5)}
		\node[draw, circle,xshift=#6,yshift=#7] (#4\nodeidx) at ({360/#1 * (\s - 1)}:#2) {$\nodeidx$};%
		\ifnum #8=1%
			\draw[->, >=latex,xshift=#6,yshift=#7] ({360/#1 * (\s - 1)+#3}:#2)%
				arc ({360/#1 * (\s - 1)+#3}:{360/#1 * (\s)-#3}:#2);%
		\else%
		\fi%
	}%
}

\DeclareMathOperator{\size}{size}
\DeclareMathOperator{\cost}{cost}
\DeclareMathOperator{\diff}{diff}
\DeclareMathOperator{\sig}{sig}

\newcommand{\cp}{\ensuremath{\to}}
\newcommand{\greedy}{\textsc{Greedy}\xspace}

\newcommand{\rtglong}{Register Transfer Graph\xspace}
\newcommand{\rtgslong}{Register Transfer Graphs\xspace}
\newcommand{\rtg}{RTG\xspace}
\newcommand{\rtgs}{RTGs\xspace}
\newcommand{\prtg}{PRTG\xspace}
\newcommand{\prtgs}{PRTGs\xspace}
\newcommand{\odrtg}{outdegree\mbox{-1} \rtg}

\newcommand{\odrtgs}{outdegree\mbox{-1} \rtgs}
\newcommand{\Odrtgs}{Outdegree-1 \rtgs}

\newcommand{\tT}{\ensuremath{\tilde{T}}}

\let\doendproof\endproof
\renewcommand{\endproof}{\hfill\qed\doendproof}

\begin{document}

\mainmatter

\title{Optimal Shuffle Code with Permutation Instructions}

\titlerunning{Optimal Shuffle Code with Permutation Instructions}

\author{Sebastian Buchwald \and Manuel Mohr \and Ignaz Rutter}

\authorrunning{S. Buchwald \and M. Mohr \and I. Rutter}

\pagestyle{plain}

\institute{Karlsruhe Institute of Technology\\
\email{\{sebastian.buchwald, manuel.mohr, rutter\}@kit.edu}}

\maketitle

\begin{abstract}
  During compilation of a program, register allocation is the task of mapping program variables to machine registers.
  During register allocation, the compiler may introduce \emph{shuffle code}, consisting of copy and swap operations, that transfers data between the registers.
  Three common sources of shuffle code are conflicting register mappings at joins in the control flow of the program, e.g, due to if-statements or loops; the calling convention for procedures, which often dictates that input arguments or results must be placed in certain registers; and machine instructions that only allow a subset of registers to occur as operands.

  Recently, Mohr et al.~\cite{mohr13cases} proposed to speed up
  shuffle code with special hardware instructions that arbitrarily
  permute the contents of up to five registers and gave a heuristic
  for computing such shuffle codes.

  In this paper, we give an efficient algorithm for generating optimal
  shuffle code in the setting of Mohr et al.  An interesting special
  case occurs when no register has to be transferred to more than one
  destination, i.e., it suffices to permute the contents of the
  registers.  This case is equivalent to factoring a permutation into
  a minimal product of permutations, each of which permutes up to five
  elements.
\end{abstract}

\section{Introduction}

One of the most important tasks of a compiler during code generation
is register allocation, which is the task of mapping program variables
to machine registers.  During this phase, it is frequently
necessary to insert so-called \emph{shuffle code} that transfers
values between registers.  Common reasons for the insertion of shuffle
code are control flow joins, procedure calling conventions and
constrained machine instructions.

The specification of a shuffle code, i.e., a description which
register contents should be transferred to which registers, can be
formulated as a directed graph whose vertices are the registers and an
edge $(u,v)$ means that the content of $u$ before the execution of the
shuffle code must be in $v$ after the execution.  Naturally, every
vertex must have at most one incoming edge.  Note that vertices may
have several outgoing edges, indicating that their contents must be
transferred to several destinations, and even loops $(u,u)$,
indicating that the content of register $u$ must be preserved.  We
call such a graph a \emph{\rtglong} or \emph{\rtg}.  Two important
special types of \rtgs are \odrtgs where the maximum out-degree is~$1$
and \prtgs where $\deg^-(v) = \deg^+(v) = 1$ for all vertices $v$
($\deg^-$ and $\deg^+$ denote the in- and out-degree of a vertex,
respectively).

We say that a shuffle code, consisting of a sequence of
copy and swap operations on the registers, \emph{implements} an \rtg
if after the execution of the shuffle code every register whose
corresponding vertex has an incoming edge has the correct content.
The \emph{shuffle code generation} problem asks for a shortest shuffle
code that implements a given \rtg.

The amount of shuffle code directly depends on the quality of copy
coalescing, a subtask of register allocation~\cite{mohr13cases}.  As copy
coalescing is NP-complete~\cite{Bouchez2007}, reducing the amount of
shuffle code is expensive in terms of compilation time, and thus
cannot be afforded in all contexts, e.g., just-in-time compilation.

Therefore, it has been suggested to allow more complicated operations
than simply copying and swapping to enable more efficient
shuffle code.  Mohr et al.~\cite{mohr13cases}
propose to allow performing permutations on the contents of small sets
of up to five registers.  The processor they develop
offers three instructions to implement shuffle code:
\begin{compactenum}
\item[\texttt{copy:}] copies the content of one register to another one
\item[\texttt{permi5:}] cyclically shifts the contents of up to five registers
\item[\texttt{permi23:}] swaps the contents of two registers and
  performs a cyclic shift of the contents of up to three registers; the
  two sets of registers must be disjoint.
\end{compactenum}

\noindent In fact, the two operations \texttt{permi5} and \texttt{permi23}
together allow to arbitrarily permute the contents of up to five
registers in a single operation.  A corresponding hardware and a
modified compiler that employs a greedy approach to generate the
shuffle code have been shown to improve performance in
practice~\cite{mohr13cases}.  While the greedy heuristic works well in
practice, it does not find an optimal shuffle code in all cases.

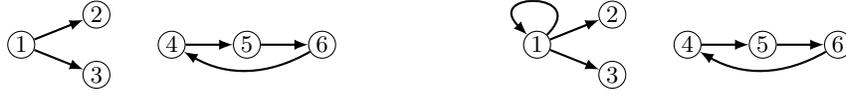
\begin{figure}[t]
	\centering
	\begin{subfigure}[t]{.46\textwidth}
		\centering
		\begin{tikzpicture}[every node/.style={draw, circle, inner sep=1pt}]
			\node (a1) at (1,0)    {$1$};
			\node (a2) at (2,0.4)  {$2$};
			\node (a3) at (2,-0.4) {$3$};
			\node (a4) at (3,0)    {$4$};
			\node (a5) at (4,0)    {$5$};
			\node (a6) at (5,0)    {$6$};

			\draw[-latex,thick]           (a1) to (a2);
			\draw[-latex,thick]           (a1) to (a3);
			\draw[-latex,thick]           (a4) to (a5);
			\draw[-latex,thick]           (a5) to (a6);
			\draw[-latex,thick,bend left] (a6) to (a4);
		\end{tikzpicture}
	\end{subfigure}
	\hfill
	\begin{subfigure}[t]{.46\textwidth}
		\centering
		\begin{tikzpicture}[every node/.style={draw, circle, inner sep=1pt}]
			\node (a1) at (1,0)    {$1$};
			\node (a2) at (2,0.4)  {$2$};
			\node (a3) at (2,-0.4) {$3$};
			\node (a4) at (3,0)    {$4$};
			\node (a5) at (4,0)    {$5$};
			\node (a6) at (5,0)    {$6$};

			\draw[-latex,thick,loop]      (a1) to (a1);
			\draw[-latex,thick]           (a1) to (a2);
			\draw[-latex,thick]           (a1) to (a3);
			\draw[-latex,thick]           (a4) to (a5);
			\draw[-latex,thick]           (a5) to (a6);
			\draw[-latex,thick,bend left] (a6) to (a4);
		\end{tikzpicture}
	\end{subfigure}
	\caption{Two example \rtgs where the optimal shuffle code is not obvious.}
	\label{fig:outline-example}
\end{figure}

It is not obvious how to generate optimal shuffle code using the three
instructions \texttt{copy}, \texttt{permi5} and \texttt{permi23} even
for small \rtgs.  In the left \rtg from \Cref{fig:outline-example}, a
naive solution would implement edges $(1, 2)$ and $(1, 3)$ using
copies and the remaining cycle $(4~5~6)$ using a \texttt{permi5}.
However, using one \texttt{permi23} to implement the cycle $(4~5~6)$
and swap registers $1$ and $2$, and then copying register $2$ to $3$
requires only two instructions.  This is legal because the contents of
register $1$ can be overwritten.  The same trick is not applicable for
the right \rtg in \Cref{fig:outline-example} because of the loop $(1,
1)$ and hence three instructions are necessary to implement that \rtg.

A maximum permutation size of~$5$ may seem arbitrary at first but is
a consequence of instruction encoding constraints.  In
each \texttt{permi} instruction, the register numbers and their
order must be encoded in the instruction word.  Hence,
$\lceil \log_2{({n \choose k} k!)} \rceil$ bits of an instruction
word are needed to be able to encode all permutations of $k$
registers out of $n$ total registers.  As many machine architectures
use a fixed size for instruction words, e.g., 32 or 64~bits, and
the operation type must also be encoded in the instruction word,
space is very limited.  In fact, for a 32~bit instruction word, $34$~is
the maximum number of registers that leave enough space for the
operation type.

\subsubsection*{Related Work.}
\label{sec:related-work}

As long as only copy and swap operations are allowed, finding an optimal
shuffle code for a given \rtg is a straightforward
task~\cite[p. 56--57]{HackThesis}.  Therefore work in the area of
compiler construction in this context has focused on coalescing
techniques that reduce the number and the size of
\rtgs~\cite{Bouchez2007,Grund2007,Hack2008,Blazy2009}.

From a theoretical point of view, the most closely related work studies
the case where the input \rtg consists of a union of disjoint directed
cycles, which can be interpreted as a permutation $\pi$.  Then, no copy
operations are necessary for an optimal shuffle code and hence the
problem of finding an optimal shuffle code using
\texttt{permi23} and \texttt{permi5} is equivalent to writing $\pi$ as
a shortest product of permutations of maximum size~$5$, where a
permutation of $n$ elements has size~$k$ if it fixes $n-k$ elements.

There has been work on writing a permutation as a product of
permutations that satisfy certain restrictions.  The factorization
problem on permutation groups from computational group
theory~\cite{seress2003} is the task of writing an element
$g$ of a permutation group as a product of given generators $S$.
Hence, an algorithm for solving the factorization problem could be
applied in our context by using all possible permutations of size $5$ or
less as the set $S$.  However, the algorithms do not guarantee
minimality of the product.  For the case that $S$ consists of all
permutations that reverse a contiguous subsequence of the elements,
known as the pancake sorting problem, it has been shown that computing
a factoring of minimum size is NP-complete~\cite{Caprara97}.

Farnoud and Milenkovic~\cite{Farnoud2012} consider a weighted version
of factoring a permutation into transpositions.  They present a
polynomial constant-factor approximation algorithm for factoring a
given permutation into transpositions where transpositions have
arbitrary non-negative costs.  For the case that the transposition
costs are defined by a path-metric, they show how to compute a
factoring of minimum weight in polynomial time.  In our problem, we
cannot assign costs to an individual transposition as its cost is
context-dependent, e.g., four transpositions whose product is a cycle
require one operation, whereas four arbitrary transpositions may
require two.

\subsubsection*{Contribution and Outline.}

In this paper, we present an efficient algorithm for generating optimal
shuffle code using the operations \texttt{copy, permi5}, and
\texttt{permi23}, or equivalently, using copy operations and
permutations of size at most~$5$.

We first prove the existence of a special type of optimal shuffle
codes whose copy operations correspond to edges of the input \rtg in
Section~\ref{sec:register-transfer-graphs}.  Removing the set of edges
implemented by copy operations from an \rtg leaves an \odrtg.

We show that the greedy algorithm proposed by Mohr et
al.~\cite{mohr13cases} finds optimal shuffle codes for \odrtgs and
that the size of an optimal shuffle code can be expressed as a
function that depends only on three characteristic numbers of the
\odrtg rather than on its structure.  Since \prtgs are a special
case of \odrtgs, this shows that \greedy is a linear-time algorithm
for factoring an arbitrary permutation into a minimum number of
permutations of size at most~$5$.

Finally, in Section~\ref{sec:general-case}, we show how to compute an
optimal set of \rtg edges that will be implemented by copy operations
such that the remaining \odrtg admits a shortest shuffle code.  This
is done by several dynamic programs for the cases that the input
\rtg is disconnected, is a tree, or is connected and contains a
(single) cycle.

\section{\rtgslong and Optimal Shuffle Codes}
\label{sec:register-transfer-graphs}

In this section, we rephrase the shuffle code generation problem as a
graph problem.  An \rtg that has only self-loops needs
no shuffle-code and is called \emph{trivial}.

It is easy to define the effect of a permutation on an \rtg.
Let $G$ be an \rtg and let $\pi$ be an arbitrary
permutation that is applied to the contents of the registers.  We
define $\pi G = (V, \pi E)$, where $\pi E = \{ (\pi(u), v) \mid (u,v)
\in E\}$.  This models the fact that if $v$ should receive the data
contained in $u$, then after $\pi$ moves the data contained in $u$ to
some other register $\pi(u)$, the data contained in $\pi(u)$ should
end up in $v$.  We observe that for two permutations $\pi_1,\pi_2$ of
$V$, it is $(\pi_2 \circ \pi_1)G = \pi_2(\pi_1(G))$, i.e., we have
defined a group action of the symmetric group on \rtgs.
For \prtgs, the shuffle code generation problem asks for a shortest shuffle
code that makes the given \prtg trivial.

Unfortunately, it is not possible to directly express copy
operations in \rtgs.  Instead, we rely on the following observation.
Consider an arbitrary shuffle code that contains a copy $a
\cp b$ with source $a$ and target $b$ that is followed by a
transposition $\tau$ of the contents of registers $c$ and $d$.  We can
replace this sequence by a transposition of the registers $\{c,d\}$
and a copy $\tau(a) \cp \tau(b)$.  Thus, given a sequence of
operations, we can successively move the copy operations to the end of
the sequence without increasing its length.  Thus, for
any \rtg there exists a shuffle code that consists of a pair of
sequences $((\pi_1,\dots,\pi_p), (c_1,\dots,c_t) )$, where the
$\pi_i$ are permutation operations and the $c_i$ are
copy operations.  We now strengthen our assumption on the copy
operations.

\begin{lemma}
  \label{lem:shuffle-normalized}
  Every instance of the shuffle code generation problem has an optimal
  shuffle code $((\pi_1,\dots,\pi_p), (c_1,\dots,c_t))$ such that
  \begin{compactenum}[(i)]
  \item No register occurs as both a source and a target of copy
    operations.
  \item Every register is the target of at most one copy operation.
  \item There is a bijection between the copy operations $c_i$ and the
    edges of $\pi G$ that are not loops, where $\pi = \pi_p \circ
    \pi_{p-1} \circ \cdots \circ \pi_1$.
  \item If $u$ is the source of a copy operation, then $u$ is incident
    to a loop in $\pi G$.
  \item The number of copies is $\sum_{v \in V} \max \{ \deg_G^+(v) - 1,
    0\}$.
  \end{compactenum}
\end{lemma}

\begin{proof}
  Consider an optimal shuffle code of the form $((\pi_1,\dots,\pi_p),
  (c_1,\dots,c_t))$ as above and assume that the number $t$ of copy
  operations is minimal among all optimal shuffle codes.

  Suppose there exists a register that occurs as both a source and a
  target of copy operations or a register that occurs as the target of
  more than one copy operation.  Let $k$ be the smallest index such
  that in the sequence $c_1,\dots,c_k$ there is a register occurring
  as both a source and a target or a register that occurs as a target
  of two copy operations.  We show that we can modify the sequence of
  copy operation such that the length of the prefix without such
  registers increases.  Inductively, we then obtain a sequence without
  such registers.  Let $v$ and $w$ denote the source and target of
  $c_k$, respectively.  Let $i$ denote the largest index such that
  $c_i$ is a copy operation that has $w$ as a source or target or such
  that $c_i$ is a copy operation with target $v$.  We distinguish
  three cases based on whether $c_i$ has target $v$, target $w$, or
  source $w$.

  \begin{figure}[tb]
    \centering
    \begin{subfigure}[b]{.3\textwidth}
      \centering
      \includegraphics[page=1]{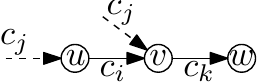}
      \caption{}
      \label{fig:copy-properties-a}
    \end{subfigure}
    \hfill
    \begin{subfigure}[b]{.3\textwidth}
      \centering
      \includegraphics[page=2]{fig/copy-properties}
      \caption{}
      \label{fig:copy-properties-b}
    \end{subfigure}
    \hfill
    \begin{subfigure}[b]{.3\textwidth}
      \centering
      \includegraphics[page=3]{fig/copy-properties}
      \caption{}
      \label{fig:copy-properties-c}
    \end{subfigure}
    \caption{Illustration of the proof of
      Lemma~\ref{lem:shuffle-normalized}.  The copies $c_j$ with
      $i<j<k$ along the dashed edges would contradict either the
      choice of $i$ or $k$.}
    \label{fig:copy-properties}
  \end{figure}

  \emph{Case 1:} The target of $c_i$ is $v$; see
  Fig.~\ref{fig:copy-properties-a}.  Let $u$ denote the source of
  operation $c_i$.  The sequence first copies a value from $u$ to $v$
  and from there to $w$.  We then replace $c_k$ by a copy with source
  $u$ and target $w$.  (If $u=w$, we omit the operation altogether.)
  This only changes the outcome of the shuffle code if the value
  contained in $u$ or $v$ is modified between operations $c_i$ and
  $c_k$, i.e., if there exists a copy operation $c_j$ with $i<j<k$
  whose target is either $u$ or $v$.  But then already the smaller
  sequence $c_1,\dots,c_j$ has $u$ occur as both a source and a target
  or $v$ as a target of two operations, contradicting the minimality
  of $k$.

  \emph{Case 2:} The target of $c_i$ is $w$; see
  Fig.~\ref{fig:copy-properties-b}.  In this case the copy operation
  $c_i$ copies a value to $w$ and later this value is overwritten by
  the operation $c_k$.  Note that by the choice of $i$ there is no
  operation $c_j$ with $i<j<k$ with source $w$.  Thus, omitting the
  copy operation $c_i$ does not change the outcome of the shuffle
  code.  A contradiction to optimality.

  \emph{Case 3:} The source of $c_i$ is $w$; see
  Fig.~\ref{fig:copy-properties-c}.  Let $x$ denote the target of
  operation $c_i$.  In this case first a value is copied from $w$ to
  $x$ and later the value in $v$ is copied to $w$.  We claim that no
  copy operation $c_j$ with $i<j<k$ involves $x$ or $w$.  If $x$
  occurs as the source of $c_j$ (as the target of $c_j$), then $x$
  occurs as a source and target (two times as a target) in the
  sequence $c_1,\dots,c_j$, contradicting the minimality of $k$.  If
  $w$ is the target of $c_j$, then $w$ occurs as a source and a target
  in the sequence $c_1,\dots,c_j$, contradicting the choice of $k$.
  If $w$ is the source of $c_j$ we have a contradiction to the choice
  of $i$.  This proves the claim.  We can thus, without changing the
  outcome of the shuffle code move the operation $c_i$ immediately
  before the operation $c_k$.  Then our sequence contains consecutive
  copy operations $w \cp x$ and $v \cp w$.  Replace these two
  operations by a cyclic shift of $v$, $w$ and $x$ and a copy
  operation $w \cp v$.  This decreases the number of copy operations
  by 1 and thus contradicts the minimality of $t$.

  Altogether, in each case, we have either found a contradiction to
  the optimality of the shuffle code, to the minimality of the number
  of copy operations or we have succeeded in producing a shuffle code
  that has a longer prefix satisfying properties (i) and (ii).
  Inductively, we obtain a shuffle code satisfying both (i) and (ii).
  Fix such a code.  Since no register is both source and target of a
  copy operation, the copy operations are commutative and can be
  reordered arbitrarily without changing the result.

  For property (iii) first observe that the only way to transfer a
  value from $u$ to $v$ is via a copy operation $u \cp v$.  This is
  due to the fact that the shuffle code is correct, that no node
  occurs as both a source and a target of copy operations, and that
  $\pi$ only permutes the values in the initial registers but does not
  duplicate them.  Thus, for every edge there must be a corresponding
  copy operation.  Conversely, this number of copy operations
  certainly suffices for a correct shuffle code for $\pi G$.

  For property (iv) consider a copy operation from $u$ to $v$ such
  that $u$ is not incident to a loop.  If the in-degree of $v$ in $\pi
  G$ were~1, then there would be an incoming edge, which would
  correspond to a copy operation with target $u$, which is not
  possible by property (i).  Thus, $u$ has in-degree~0.  But then, the
  contents of $u$ are irrelevant and we can replace the copy from $u$
  to $v$ by an operation that swaps the contents of $u$ and $v$,
  resulting in a shuffle code with fewer copy operations.

  By property (iv) every vertex that is the source of an edge in $\pi
  G$ is incident to a loop.  Hence $\sum_{v \in V} \max\{\deg_{\pi
    G}^+(v)-1,0\}$ is the number of non-loop edges in $\pi G$, which
  is the same as the number of copy operations by property (iii).
  Note that by definition $\pi$ only permutes the out-degrees of the
  vertices, and hence $\sum_{v \in V} \max\{\deg_{\pi G}^+(v)-1,0\} =
  \sum_{v \in V} \max\{\deg_{G}^+(v)-1,0\}$.  This shows property (iv)
  and finishes the proof.
\end{proof}

We call a shuffle code satisfying the conditions of
Lemma~\ref{lem:shuffle-normalized} \emph{normalized}.  Observe that
the number of copy operations used by a normalized shuffle code is a
lower bound on the number of necessary copy operations since
permutations, by definition, only permute values but never create
copies of them.

Consider now an \rtg $G$ together with a normalized optimal shuffle code and one of its copy operations $u \cp v$.
Since the code is normalized, the value transferred to $v$ by this copy operation is the one that stays there after the shuffle code has been executed.
If $v$ had no incoming edge in $G$, then we could shorten the shuffle by omitting the copy operation.
Thus, $v$ has an incoming edge $(u',v)$ in $G$, and we associate the copy $u \cp v$ with the edge $(u',v)$ of $G$.
In fact, $u' = \pi^{-1}(u)$, where $\pi = \pi_p \circ \cdots \circ \pi_1$.
In this way, we associate every copy operation with an edge of the input \rtg.
In fact, this is an injective mapping by \Cref{lem:shuffle-normalized} (ii).

\begin{lemma}
  \label{lem:cut-copies}
  Let $((\pi_1,\dots,\pi_p), (c_1,\dots,c_t))$ be an optimal shuffle code $S$
  for an \rtg $G=(V,E)$ and let $C \subseteq E$ be the edges
  that are associated with copies in $S$.  Then

  \begin{compactenum}[(i)]
  \item Every vertex $v$ has $\max\{\deg_{G}^+(v) - 1, 0\}$ outgoing edges in
      $C$.
  \item $G-C$ is an \odrtg.
  \item $\pi_1,\dots,\pi_p$ is an optimal shuffle code for $G-C$.
  \end{compactenum}
\end{lemma}

\begin{proof}
  For property (i) observe that, since permuting the register contents
  does not duplicate values, it is necessary that at least
  $\max\{\deg_{G}^+(v)-1,0\}$ of the edges of $v$ are implemented by
  copy operations and thus are in $C$.  By property (v) of
  Lemma~\ref{lem:shuffle-normalized} the number of copy operations is
  exactly the sum of these values, which immediately implies that
  equality holds at every vertex.

  Property (ii) follows immediately from property (i).

  Finally, for property (iii), suppose there is a shorter optimal
  shuffle code $\pi_1',\dots,\pi_{p'}'$ with $p' < p$ for $G-C$.  Let
  $\pi' = \pi_{p'}' \circ \cdots \circ \pi_1'$.  Then $\pi' G$ has
  $|C|$ edges that are not loops and by creating a copy operation for
  each of them we obtain a shorter shuffle code.  This is a
  contradiction to the optimality of the original shuffle code.  Hence
  property (iii) holds.
\end{proof}

Lemma~\ref{lem:cut-copies} shows that an optimal shuffle code for an
\rtg $G$ can be found by first picking for each vertex one of
its outgoing edges (if it has any) and removing the remaining edges from
$G$, second finding an optimal shuffle code for the resulting \odrtg, and
finally creating one copy operation for each of the previously removed
edges. Fig.~\ref{fig:copy-removal} shows that the choice of the
outgoing edges is crucial to obtain an optimal shuffle code.

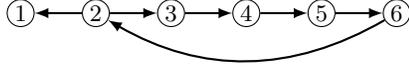
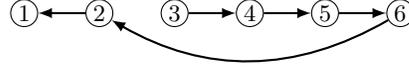
\begin{figure}[t]
	\centering
	\begin{subfigure}[t]{.46\textwidth}
		\centering
		\begin{tikzpicture}[every node/.style={draw, circle, inner sep=1pt}]
			\node (a1) at (1,0) {$1$};
			\node (a2) at (2,0) {$2$};
			\node (a3) at (3,0) {$3$};
			\node (a4) at (4,0) {$4$};
			\node (a5) at (5,0) {$5$};
			\node (a6) at (6,0) {$6$};

			\draw[-latex,thick]           (a2) to (a1);
			\draw[-latex,thick]           (a2) to (a3);
			\draw[-latex,thick]           (a3) to (a4);
			\draw[-latex,thick]           (a4) to (a5);
			\draw[-latex,thick]           (a5) to (a6);
			\draw[-latex,thick,bend left] (a6) to (a2);
		\end{tikzpicture}
		\caption{The original \rtg $G$ needs one permutation and one copy operation.}
	\end{subfigure}
	\hfill
	\begin{subfigure}[t]{.46\textwidth}
		\centering
		\begin{tikzpicture}[every node/.style={draw, circle, inner sep=1pt}]
			\node (a1) at (1,0) {$1$};
			\node (a2) at (2,0) {$2$};
			\node (a3) at (3,0) {$3$};
			\node (a4) at (4,0) {$4$};
			\node (a5) at (5,0) {$5$};
			\node (a6) at (6,0) {$6$};

			\draw[-latex,thick]           (a2) to (a1);
			\draw[-latex,thick]           (a3) to (a4);
			\draw[-latex,thick]           (a4) to (a5);
			\draw[-latex,thick]           (a5) to (a6);
			\draw[-latex,thick,bend left] (a6) to (a2);
		\end{tikzpicture}
		\caption{After removing the edge $(2,3)$, the \rtg needs two permutation operations.}
	\end{subfigure}
	\caption{The \rtg $G$ obtains the normalized optimal shuffle code ($\pi_1$, $c_1$), where $\pi_1 = (23456)$ and $c_1=3 \cp 1$.
	         However, after removing the edge $(2,3)$ (instead of $(1,2)$) we cannot achieve an optimal solution anymore.}
	\label{fig:copy-removal}
\end{figure}

In the following, we first show how to compute an optimal shuffle code
for an \odrtg in \Cref{sec:path-cycle-rtgs}.  Afterwards, in
\Cref{sec:general-case}, we design an algorithm for efficiently
determining a set of edges to be removed such that the resulting
\odrtg admits a shuffle code with the smallest number of operations.

\section{Optimal Shuffle Code for \Odrtgs}
\label{sec:path-cycle-rtgs}

In this section we prove the optimality of the greedy algorithm
proposed by Mohr et al.~\cite{mohr13cases} for \odrtgs.  Before we
formulate the algorithm, let us look at the effect of applying a
transposition $\tau = (u~v)$ to contiguous vertices of a $k$-cycle
$K = (V_K, E_K)$ in a \prtg $G$, where $k$-cycle denotes a cycle
of size $k$.
Hence, $u, v \in V_K$ and $(u, v) \in E_K$.  Then, in $\tau G$,
the cycle $K$ is replaced by a
$(k-1)$-cycle and a vertex $v$ with a loop.  We say that
$\tau$ has reduced the size of $K$ by~$1$.  If $\tau K$ is trivial, we
say that $\tau$ resolves $K$.  It is easy to see that \texttt{permi5}
reduces the size of a cycle by up to $4$ and \texttt{permi23}
reduces the sizes of two distinct cycles by~$1$ and up to~$2$, respectively.
We can now formulate \greedy as follows.
\begin{compactenum}
\item Complete each directed path of the input \odrtg into a directed cycle,
  thereby turning the input into a \prtg.
\item While there exists a cycle $K$ of size at least~$4$, apply a
  \texttt{permi5} operation to reduce the size of $K$ as much as
  possible.
\item While there exist a 2-cycle and a 3-cycle,
  resolve them with a \texttt{permi23} operation.
\item Resolve pairs of 2-cycles by \texttt{permi23} operations.
\item Resolve triples of 3-cycles by pairs of \texttt{permi23}
  operations.
\end{compactenum}

We claim that \greedy computes an optimal shuffle code.  Let $G$ be an
\odrtg and let $Q$ denote the set of paths and cycles of $G$.  For a path or cycle
$\sigma \in Q$, we denote by $\size(\sigma)$ the number of vertices of
$\sigma$.  Define $X = \sum_{\sigma \in Q} \lfloor \size(\sigma)/4
  \rfloor$ and $a_i = | \{ \sigma \in Q \mid \size(\sigma) = i \mod
  4\}|$ for $i=2,3$.  We call the triple $\sig(G) = (X,a_2,a_3)$ the
  \emph{signature} of $G$.

\begin{lemma}
  \label{lem:greedy-cost}
  Let $G$ be an \odrtg with $\sig(G) = (X,a_2,a_3)$.
  The number $\greedy(G)$ of operations in the shuffle code produced
  by the greedy algorithm is
  $\greedy(G) = X + \max\{\lceil (a_2+a_3)/2 \rceil, \lceil (a_2 + 2
  a_3)/3 \rceil\}$.
\end{lemma}

\begin{proof}
  After the first step we have a \prtg with the same signature as $G$.
  Clearly, \greedy produces exactly $X$ operations for reducing all
  cycle sizes below 4.  Afterwards, only \texttt{permi23} operations
  are used to resolve the remaining cycles of size~2 and~3.

  If $a_2 \ge a_3$, then first $a_3$ operations are used to resolve
  pairs of cycles of size~2 and~3.  Afterwards, the remaining $a_2-a_3$
  cycles of size~2 are resolved by using $\lceil (a_2 - a_3)/2 \rceil$
  operations.  In total, these are $\lceil (a_2 + a_3)/2 \rceil$
  operations.

  If $a_3 \ge a_2$, then first $a_2$ operations are used to resolve
  pairs of cycles of size~2 and~3.  Afterwards, the remaining $a_3 -
  a_2$ cycles of size~3 are resolved by using $\lceil 2(a_3 - a_2)/3
  \rceil$ operations.  In total, these are $\lceil (a_2 + 2a_3)/3
  \rceil$ operations.

  We observe that $(a_2 + a_3)/2 \le (a_2 + 2a_3)/3$ holds if and only
  if $a_2 \le a_3$ and that equality holds for $a_2 = a_3$.  Since
  $\lceil \cdot \rceil$ is a monotone function, this implies that the
  total cost produced by the last part of the algorithm is $\max
  \{\lceil (a_2 + a_3)/2 \rceil, \lceil (a_2 + 2a_3)/3 \rceil \}$.
\end{proof}

In particular, the length of the shuffle code computed by \greedy only
depends on the signature of the input \rtg $G$.  In the remainder of
this section, we prove that \greedy is optimal for
\odrtgs and therefore the formula in Lemma~\ref{lem:greedy-cost} actually
computes the length of an optimal shuffle code.

\begin{lemma}
  \label{lem:indicator-function}
  Let $G, G'$ be \prtgs with $\sig(G) = (X,a_2,a_3)$,
  $\sig(G') = (X',a_2',a_3')$ and $\greedy(G) - \greedy(G')
  \ge c$, and let $(\Delta_X,\Delta_2,\Delta_3) = \sig(G) - \sig(G')$.
  If $a_2 \ge a_3$, then $2\Delta_X + \Delta_2 + \Delta_3 \le -2c + 1$.
  If $a_3 > a_2$, then $3\Delta_X + \Delta_2 + 2\Delta_3 \le -3c + 2$.
\end{lemma}

\begin{proof}
  We assume that $\greedy(G) - \greedy(G') \ge c$ and start with the
  case that $a_2 \ge a_3$.  By Lemma~\ref{lem:greedy-cost} and basic
  calculation rules for $\lceil \cdot \rceil$, we have the following.
  \begin{align*}
    \greedy(G) &= X + \lceil (a_2 + a_3)/2 \rceil \le X+ (a_2+a_3+1)/2\\
    \greedy(G') &\ge X' + \lceil (a_2' + a_3')/2 \rceil \ge X+\Delta_X
    + (a_2 + a_3 + \Delta_2 + \Delta_3)/2
  \end{align*}
  Therefore, their difference computes to
  \begin{align*}
    \greedy(G) - \greedy(G') &\le -\Delta_X - (\Delta_2+\Delta_3-1)/2\\
    &= - (2\Delta_X + \Delta_2+\Delta_3 -1)/2.
  \end{align*}
  By assumption, we thus have $- (2\Delta_X + \Delta_2+\Delta_3 -1)/2
  \ge c$, or equivalently $2\Delta_X + \Delta_2 + \Delta_3 \le -2c +
  1$.

  Now consider the case $a_3 > a_2$.  By Lemma~\ref{lem:greedy-cost},
  we have the following.
  \begin{align*}
    \greedy(G) &= X + \lceil (a_2 + 2a_3)/3 \rceil \le X+ (a_2+2a_3+2)/3\\
    \greedy(G') &\ge X' + \lceil (a_2' + 2a_3')/3 \rceil \ge
    X+\Delta_X + (a_2 + 2a_3 + \Delta_2 + 2\Delta_3)/3
  \end{align*}
  Similar to above, their difference computes to
  \begin{align*}
    \greedy(G) - \greedy(G') &\le -\Delta_X -
    (\Delta_2+2\Delta_3-2)/3\\
    &= - (3\Delta_X + \Delta_2+2\Delta_3 -2)/3.
  \end{align*}
  Similarly as above, by assumption we have $- (3\Delta_X +
  \Delta_2+2\Delta_3 -2)/3 \ge c$, which is equivalent to $3\Delta_X +
  \Delta_2 +2 \Delta_3 \le -3c+2$.
\end{proof}

Lemma~\ref{lem:indicator-function}
gives us necessary conditions for
when the \greedy solutions of two \rtgs differ by some value $c$.
These necessary conditions depend only on the difference of the two
signatures.  To study them more precisely, we define
$\Psi_1(\Delta_X,\Delta_2,\Delta_3) = 2\Delta_X + \Delta_2 + \Delta_3$
and $\Psi_2(\Delta_X, \Delta_2, \Delta_3) = 3\Delta_X + \Delta_2 + 2
\Delta_3$.  Next, we study the effect of a single transposition on
these two functions.

Let $G = (V,E)$ be a \prtg with $\sig(G) = (X,a_2,a_3)$ and let $\tau$
be a transposition of two elements in $V$.  We distinguish cases based
on whether the swapped elements are in different connected components
or not.  In the former case, we say that $\tau$ is a \emph{merge}, in
the latter we call it a \emph{split}; see
Fig.~\ref{fig:split-merge-illustration} for an illustration.

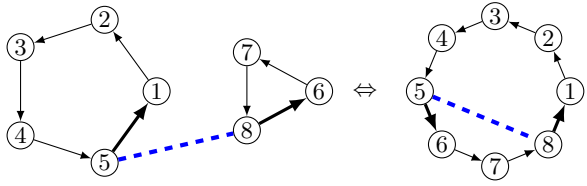
\begin{figure}[tb]
	\centering
	\begin{tikzpicture}[every node/.style={draw, circle, inner sep=1pt}]
		\node (a1) at (0:1cm) {$1$};
		\node (a2) at (72:1cm) {$2$};
		\node (a3) at (144:1cm) {$3$};
		\node (a4) at (216:1cm) {$4$};
		\node (a5) at (288:1cm) {$5$};

		\draw[->, -latex] (a1) to (a2);
		\draw[->, -latex] (a2) to (a3);
		\draw[->, -latex] (a3) to (a4);
		\draw[->, -latex] (a4) to (a5);
		\draw[very thick, ->, -latex] (a5) to (a1);

		\node[xshift=2.5cm] (b1) at (0:0.66cm) {$6$};
		\node[xshift=2.5cm] (b2) at (120:0.6cm) {$7$};
		\node[xshift=2.5cm] (b3) at (240:0.6cm) {$8$};

		\draw[->, -latex] (b1) to (b2);
		\draw[->, -latex] (b2) to (b3);
		\draw[very thick, ->, -latex] (b3) to (b1);

		\draw[ultra thick, blue, dashed] (a5) to (b3);

		\node[draw=none, right=2mm of b1] (arr) {$\Leftrightarrow$};

		\node[xshift=5.5cm] (c1) at (0:1cm) {$1$};
		\node[xshift=5.5cm] (c2) at (45:1cm) {$2$};
		\node[xshift=5.5cm] (c3) at (90:1cm) {$3$};
		\node[xshift=5.5cm] (c4) at (135:1cm) {$4$};
		\node[xshift=5.5cm] (c5) at (180:1cm) {$5$};
		\node[xshift=5.5cm] (c6) at (225:1cm) {$6$};
		\node[xshift=5.5cm] (c7) at (270:1cm) {$7$};
		\node[xshift=5.5cm] (c8) at (315:1cm) {$8$};

		\draw[->, -latex] (c1) to (c2);
		\draw[->, -latex] (c2) to (c3);
		\draw[->, -latex] (c3) to (c4);
		\draw[->, -latex] (c4) to (c5);
		\draw[very thick, ->, -latex] (c5) to (c6);
		\draw[->, -latex] (c6) to (c7);
		\draw[->, -latex] (c7) to (c8);
		\draw[very thick, ->, -latex] (c8) to (c1);

		\draw[ultra thick, blue, dashed] (c5) to (c8);
	\end{tikzpicture}
	\caption{The transposition $\tau = (5~8)$ acting on \prtgs. Affected edges are drawn thick. Read from left to right, the transposition is a merge; read from right to left, it is a split.}
	\label{fig:split-merge-illustration}
\end{figure}

We start with the merge operations as they are a bit simpler.  When
merging two cycles of size $s_1$ and $s_2$, respectively, they are
replaced by a single cycle of size $s_1+s_2$.  Note that removing the
two cycles may decrease the values $a_2$ and $a_3$ of the signature by
at most~$2$ in total.  On the other hand, the new cycle can potentially
increase one of these values by~$1$.  The value $X$ never decreases, and
it increases by~$1$ if and only if $s_1 \mod 4 + s_2 \mod 4 \ge 4$.
Table~\ref{tab:merge-effects} shows the possible signature changes
$(\Delta_X, \Delta_2,\Delta_3)$ resulting from a merge.  The entry
in row $i$ and column $j$ shows the result of merging two cycles whose
sizes modulo $4$ are $i$ and $j$, respectively.  \Cref{tab:merge-costs}
shows the corresponding values of $\Psi_1$ and $\Psi_2$.
Only entries with $i \le j$ are shown, the remaining cases are symmetric.

\begin{table}[tb]
	\centering
	\begin{subfigure}[b]{.5\textwidth}
		\centering
		\begin{tabular}{c||c|c|c|c}
				& $0$         & $1$         & $2$          & $3$ \\
			\hline
			\hline
			$0$ & $(0, 0, 0)$ & $(0, 0, 0)$ & $(0, 0, 0)$  & $(0, 0, 0)$ \\
			\hline
			$1$ &             & $(0, 1, 0)$ & $(0, -1, 1)$ & $(1, 0, -1)$ \\
			\hline
			$2$ &             &             & $(1, -2, 0)$ & $(1, -1, -1)$ \\
			\hline
			$3$ &             &             &              & $(1, 1, -2)$ \\
		\end{tabular}
		\caption{Signature change $(\Delta_X, \Delta_2,
                  \Delta_3)$.}
		\label{tab:merge-effects}
	\end{subfigure}
        \hfill
	\begin{subfigure}[b]{0.45\textwidth}
          \centering
		\begin{tabular}{c||c|c|c|c}
			& $0$ & $1$ & $2$ & $3$ \\
			\hline
			\hline
			$0$ & $0$ & $0$ & $0$ & $0$ \\
			\hline
			$1$ &     & $1$ & $0$ & $1$ \\
			\hline
			$2$ &     &     & $0$ & $0$ \\
			\hline
			$3$ &     &     &     & $1$ \\
		\end{tabular}
                \qquad\qquad
		\begin{tabular}{c||c|c|c|c}
				& $0$ & $1$ & $2$ & $3$ \\
			\hline
			\hline
			$0$ & $0$ & $0$ & $0$ & $0$ \\
			\hline
			$1$ &     & $1$ & $1$ & $1$ \\
			\hline
			$2$ &     &     & $1$ & $0$ \\
			\hline
			$3$ &     &     &     & $0$ \\
		\end{tabular}
		\caption{Values of $\Psi_1$ (left) and $\Psi_2$ (right).}
		\label{tab:merge-costs}
	\end{subfigure}
	\caption{Signature changes and $\Psi$ values for merges.  Row
          and column are the cycle sizes modulo $4$ before the merge.}
	\label{tab:merges}
\end{table}

\begin{lemma}
  \label{lem:merge-effect}
  Let $G$ be a \prtg with $\sig(G) = (X,a_2,a_3)$ and let $\tau$
  be a merge.  Then $\greedy(G) \le \greedy(\tau G)$.
\end{lemma}

\begin{proof}
  Suppose $\greedy(\tau G) < \greedy(G)$.  Then $\greedy(G) -
  \greedy(\tau G) \ge 1$ and by Lemma~\ref{lem:indicator-function}
  either $\Psi_1 \le -1$ or $\Psi_2 \le -1$.  However,
  Table~\ref{tab:merge-costs} shows the values of $\Psi_1$ and
  $\Psi_2$ for all possible merges.  In all cases it is $\Psi_1,
  \Psi_2 \ge 0$.  A contradiction.
\end{proof}

In particular, the lemma shows that merges never decrease the cost of
the greedy solution, even if they were for free.  We now make a
similar analysis for splits.  It is, however, obvious that splits
indeed may decrease the cost of greedy solutions.  In fact, one can
always split cycles in a \prtg until it is trivial.

First, we study again the effect of splits on the signature change
$(\Delta_X, \Delta_2,\Delta_3)$.  Since a split is an inverse of a
merge, we can essentially reuse Table~\ref{tab:merge-effects}.  If
merging two cycles whose sizes modulo $4$ are $i$ and $j$, respectively,
results in a signature change of $(\Delta_X,\Delta_2,\Delta_3)$, then,
conversely, we can split a cycle whose size modulo $4$ is $i+j$ into
two cycles whose sizes modulo $4$ are $i$ and $j$, respectively, such
that the signature change is $(-\Delta_X, -\Delta_2, -\Delta_3)$,
and vice versa.  Note that given a cycle whose size
modulo $4$ is $s$ one has to look at all cells $(i,j)$ with
$i+j \equiv s~(\mathrm{mod~}4)$
to consider all the possible signature changes.  Since $\Psi_1, \Psi_2$
are linear, negating the signature change also negates the
corresponding value.  Thus, we can reuse \Cref{tab:merge-costs}
for splits by negating each entry.

\begin{lemma}
  \label{lem:cycle-decrease-psi}
  Let $G = (V,E)$ be a \prtg and let $\pi$ be a cyclic shift of $c$
  vertices in $V$.  Let further $(\Delta_X,\Delta_2,\Delta_3)$ be the
  signature change affected by $\pi$.  Then
  $\Psi_1(\Delta_X,\Delta_2,\Delta_3) \ge - \lceil (c-1)/2 \rceil$ and
  $\Psi_2(\Delta_X,\Delta_2, \Delta_3) \ge - \lceil (3c-3)/4 \rceil$.
\end{lemma}

\begin{proof}
  We can write $\pi = \tau_{c-1} \circ \cdots \circ \tau_{1}$ as a
  product of $c-1$ transpositions such that any two consecutive
  transpositions $\tau_i$ and $\tau_{i+1}$ affect a common element for
  $i=1,\dots,c-1$.

  Each transposition decreases $\Psi_1$ (or $\Psi_2$) by at most $1$,
  but a decrease happens only for certain split operations.  However,
  it is not possible to reduce $\Psi_1$ (or $\Psi_2$) with every
  single transposition since for two consecutive splits the second has
  to split one of the connected components resulting from the previous
  splits.  To get an overview of the sequences of splits that reduce
  the value of $\Psi_1$ (or of $\Psi_2$) by~$1$ for each split, we
  consider the following transition graphs $T_k$ for $\Psi_k$
  ($k=1,2$) on the vertex set $S =\{0,1,2,3\}$.  In the graph $T_k$
  there is an edge from $i$ to $j$ if there is a split that splits a
  component of size $i \mod 4$ such that one of the resulting components has
  size $j \mod 4$ and this split decreases $\Psi_k$ by $1$.  The transition
  graphs $T_1$ and $T_2$ are shown in Fig.~\ref{fig:transition}.

\begin{figure}[tb]
	\centering
	\begin{subfigure}[b]{0.45\textwidth}
		\centering
		\begin{tikzpicture}
			\node[circle] (0) {$0$};
			\node[circle, right=of 0] (1) {$1$};
			\node[circle, below=of 1] (2) {$2$};
			\node[circle, left=of 2] (3) {$3$};

			\draw[-latex] (2) to (1);
			\draw[-latex] (0) to (1);
			\draw[-latex] (0) to (3);
			\draw[-latex] (2) to (3);
		\end{tikzpicture}
	\end{subfigure}%
	\begin{subfigure}[b]{0.45\textwidth}
		\centering
		\begin{tikzpicture}
			\node[circle] (0) {$0$};
			\node[circle, right=of 0] (1) {$1$};
			\node[circle, below=of 1] (2) {$2$};
			\node[circle, left=of 2] (3) {$3$};

			\draw[-latex] (2) to (1);
			\draw[-latex] (3) to (2);
			\draw[-latex] (3) to (1);
			\draw[-latex] (0) to (1);
			\draw[-latex] (0) to (3);
			\draw[-latex] (0) to (2);
		\end{tikzpicture}
	\end{subfigure}
	\caption{Transition graphs for $\Psi_1$ (left) and $\Psi_2$ (right).}
        \label{fig:transition}
\end{figure}
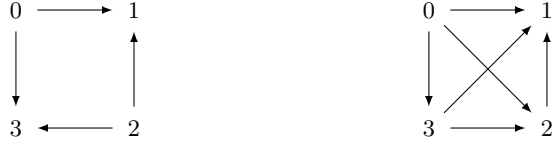

For $\Psi_1$ the longest path in the transition graph has length~$1$.
Thus, the value of $\Psi_1$ can be reduced at most every second
transposition and $\Psi_1(\Delta_X,\Delta_2,\Delta_3)
\ge - \lceil (c-1)/2 \rceil$.

For $\Psi_2$ the longest path has length~$3$ (vertex $1$ has out-degree $0$).
Therefore, after at most three consecutive steps that
decrease $\Psi_2$, there is one that does not.  It follows that at
least $\lfloor (c-1)/4 \rfloor$ operations do not decrease $\Psi_2$,
and consequently at most $\lceil (3c-3)/4 \rceil$ operations decrease
$\Psi_2$ by $1$.  Thus, $\Psi_2(\Delta_X, \Delta_2, \Delta_3) \ge -
\lceil (3c-3)/4 \rceil$.
\end{proof}

Since \texttt{permi5} performs a single cyclic shift and
\texttt{permi23} is the concatenation of two cyclic shifts,
Lemmas~\ref{lem:cycle-decrease-psi} and~\ref{lem:indicator-function}
can be used to show that no such
operation may decrease the number of operations \greedy has to perform
by more than~$1$.

\begin{corollary}
  \label{cor:one-op-greedy}
  Let $G$ be a \prtg and let $\pi$ be an operation, i.e., either a
  \emph{\texttt{permi23}} or a \emph{\texttt{permi5}}.  Then $\greedy(G) \le
  \greedy(\pi G) + 1$.
\end{corollary}

\begin{proof}
  Assume for a contradiction that $\greedy(G) > \greedy(\pi G) - 1$.
  By Lemma~\ref{lem:indicator-function} we have that either
  $\Psi_1(\Delta_X, \Delta_2,\Delta_3) \le -3$ or $\Psi_2(\Delta_X,
  \Delta_2, \Delta_3) \le -4$.

  We distinguish cases based on whether $\pi$ is a \texttt{permi5} or
  a \texttt{permi23}.  If $\pi$ is a \texttt{permi5}, then it is a
  $c$-cycle with $c \le 5$.  By Lemma~\ref{lem:cycle-decrease-psi}, we
  have that $\Psi_1(\Delta_X, \Delta_2, \Delta_3) \ge -2$ and
  $\Psi_2(\Delta_X, \Delta_2, \Delta_3) \ge -3$.  This contradicts the
  above bounds from Lemma~\ref{lem:indicator-function}.

  If $\pi$ is a \texttt{permi23}, then it is a composition of a
  2-cycle and a $c$-cycle with $c \le 3$.  According to
  Lemma~\ref{lem:cycle-decrease-psi}, both cycles contribute at least
  $-1$ to $\Psi_1$, and at least $-1$ and $-2$ to $\Psi_2$.
  Therefore, we have $\Psi_1(\Delta_X, \Delta_2, \Delta_3) \ge -2$ and
  $\Psi_2(\Delta_X, \Delta_2, \Delta_3) \ge -3$.  This is again a
  contradiction.
\end{proof}

Using this corollary and an induction on the length of an optimal
shuffle code, we show that \greedy is optimal for \prtgs;
if no operation reduces the number of operations \greedy needs
by more than~$1$, why not use the operation suggested by \greedy?

\begin{theorem}
  \label{thm:greedy-optimal-prtg}
  Let $G$ be a \prtg.  An optimal shuffle code for $G$ takes
  $\greedy(G)$ operations.  Algorithm \greedy computes
  an optimal shuffle code in linear time.
\end{theorem}

\begin{proof}
  The proof is by induction on the overall length of an optimal
  shuffle code.  Clearly, \greedy computes optimal shuffle codes for
  all instances that have a shuffle code of length~0.

  Assume that $G$ admits an optimal shuffle code of length~$k+1$.  We
  show that $\greedy(G) = k+1$.  First of all, note that $\greedy(G)
  \ge k+1$ as it computes a shuffle code of length $\greedy(G)$.  Let
  $\pi_1,\dots,\pi_{k+1}$ be a shuffle code for $G$.  Then obviously
  $\pi_{k+1}G$ admits an optimal shuffle code of length $k$, and
  therefore $\greedy(\pi_{k+1}G) = k$ by our inductive assumption.
  Corollary~\ref{cor:one-op-greedy} implies $\greedy(G) \le
  \greedy(\pi_{k+1} G) + 1 = k+1$; the induction hypothesis is proved.

  Clearly, algorithm \greedy indeed computes a correct, and thus
  optimal, shuffle code.  It can easily be implemented to run in
  linear time.
\end{proof}

Moreover, since merge operations may not decrease the cost of \greedy
and any \prtg that can be formed from the original \odrtg $G$ by
inserting edges can be obtained from the \prtg $G'$ formed by \greedy
and a sequence of merge operations, it follows that the length of an
optimal shuffle for $G$ is $\greedy(G')$.

\begin{lemma}
  \label{lem:union-prtg-optimal}
  Let $G$ be an \odrtg and let $G'$ be the \prtg formed by completing
  each directed path into a directed cycle.  Then the length of an
  optimal shuffle code of $G$ is $\greedy(G')$.
\end{lemma}

\begin{proof}
  Assume $\pi_1,\dots,\pi_k$ is an optimal shuffle code for $G$.  Of
  course, applying $\pi = \pi_k \circ \cdots \circ \pi_1$ to $G$ maps
  every value of $G$ somewhere, that is, $\pi_1,\dots,\pi_k$ is
  actually an optimal shuffle code for some instance $G''$ that
  consists of a disjoint union of directed cycles and contains $G$ as
  a subgraph.  It is not hard to see that $G''$ can be obtained from
  $G'$ by a sequence of merge operations $\tau_1,\dots,\tau_t$, i.e.,
  $G'' = \tau_t \circ \cdots \circ \tau_1 G'$.
  Lemma~\ref{lem:merge-effect} implies that $\greedy(G') \le
  \greedy(\tau_1 G') \le \cdots \le \greedy(\tau_t \circ \cdots \circ \tau_1 G')
  = \greedy(G'') = k$, where the last equality follows from
  Theorem~\ref{thm:greedy-optimal-prtg}, the optimality of \greedy for
  \prtgs.
\end{proof}

By combining Theorem~\ref{thm:greedy-optimal-prtg} and
Lemma~\ref{lem:union-prtg-optimal}, we obtain the main result of this
section.

\begin{theorem}
  \label{thm:greedy-correct}
  Let $G$ be an \odrtg.  Then an optimal shuffle code for $G$ requires
  $\greedy(G)$ operations.  \greedy computes such a shuffle code in
  linear time.
\end{theorem}

\section{The General Case}
\label{sec:general-case}

In this section we study the general case.  A \emph{copy set} of an
\rtg $G=(V,E)$ is a set $C \subseteq E$ such that $G-C = (V,E-C)$ is
an \odrtg and $|C| = \sum_{v \in V} \max\{\deg^+(v) -1, 0\}$.  We
denote by $\mathcal C(G)$ the set of all copy sets of $G$.  According
to Lemma~\ref{lem:cut-copies} an optimal shuffle code for $G$ can be
found by finding a copy set $C \in \mathcal C(G)$ such that the \odrtg
$G-C$ admits a shortest shuffle code.  By
Theorem~\ref{thm:greedy-correct} an optimal shuffle code for $G-C$ can
be computed with the greedy algorithm and its length can be
computed according to Lemma~\ref{lem:greedy-cost}.

We thus seek a
copy set $C \in \mathcal C(G)$ that minimizes the cost function
$\greedy(G - C) = X + \max\{\lceil (a_2+a_3)/2 \rceil, \lceil (a_2 +
2a_3)/3 \rceil\}$, where $(X, a_2, a_3)$ is the signature of $G - C$.
Such a copy set is called \emph{optimal}.  Clearly, this is
equivalent to minimizing the function
\[
  \greedy'(G - C) = X + \max \{ \frac{a_2 + a_3}{2}, \frac{a_2+2a_3}{3} \}
                  = \begin{cases}
                      X + \frac{a_2}{2} + \frac{a_3}{2} & \mathrm{if~}a_2 \ge a_3 \\
                      X + \frac{a_2}{3} + \frac{2a_3}{3} & \mathrm{if~}a_2 < a_3
                    \end{cases}
\]
To keep track of which case is used for evaluating $\greedy'$, we
define $\diff(G - C) = a_2 - a_3$ and compute for each of the two
function parts and every possible value $d$ a copy set $C_d$ with
$\diff(G - C_d) = d$ that minimizes that function.

More formally, we
define $\cost^1(G - C) = X + \frac{1}{2}a_2 + \frac{1}{2}a_3$ and
$\cost^2(G - C) = X + \frac{1}{3}a_2 + \frac{2}{3}a_3$ and we seek two
tables $T_G^1[\cdot], T_G^2[\cdot]$, such that $T^i_G[d]$ is the
smallest cost $\cost^i(G-C)$ that can be achieved with a copy set $C
\in \mathcal C(G)$ with $\diff(G-C) = d$.
We observe that $T^i_G[d] = \infty$ for $d < -n$ and for $d > n$.
The following lemma shows that the length of an optimal shuffle code can be
computed from these two tables.

\begin{lemma}
	\label{lem:graph-optimal}
	Let $G = (V, E)$ be an \rtg.  The length of an optimal shuffle code for~$G$ is
	$
        \sum_{v \in V} \max\{\deg^+(v) - 1,0\} + \min \{ \min_{d \ge 0}
        \lceil T_G^1[d] \rceil, \min_{d < 0} \lceil T_G^2[d] \rceil
        \}.
	$
\end{lemma}

\begin{proof}
  Let $m = \sum_{v \in V} \max\{\deg^+(v) - 1,0\}$.
  Consider an optimal normalized shuffle code for $G$, which, according
  to Lemma~\ref{lem:cut-copies}, consists of a copy set $C \subseteq
  E$ and a sequence of $k$ permutation operations, i.e., the length of
  the shuffle code is $m + k$.  Let $(X,a_2,a_3)$ denote the
  signature of $G-C$ and let $d = a_2 - a_3$.  If $a_2 \ge a_3$, or
  equivalently $d \ge 0$, then according to
  Theorem~\ref{thm:greedy-correct}, we have $k = \greedy(G) = X +
  \lceil (a_2 + a_3)/2 \rceil = \lceil X + (a_2+a_3)/2 \rceil = \lceil
  \cost^1(G-C) \rceil$, and therefore the length of the shuffle code
  is at most $m + T_G^1[d]$.  If $a_2 < a_3$, i.e., if $d < 0$,
  then we have $k = \greedy(G) = X + \lceil (a_2 + 2a_3)/3 \rceil =
  \lceil X + (a_2 + 2a_3)/3 \rceil = \lceil \cost^2(G-C) \rceil$, and
  therefore the length of the shuffle code is at most $m +
  T_G^2[d]$.  In either case the length of the shuffle code is bounded
  by the expression given in the statement of the theorem.

  Conversely, assume that the minimum of the expression is obtained
  for some value $T_G^i[d]$.  If $d < 0$ (resp. if $d \ge 0$), there
  exists a copy set $C$ such that $\sig(G-C) = (X,a_2,a_3)$ and
  $\greedy(G-C) = \lceil \cost^2(G-C) \rceil$ (resp. $\greedy(G-C) =
  \lceil \cost^2(G-C) \rceil$) is at most $T_G^1[d]$ (resp. at most
  $T_G^2[d]$).  Then, clearly, the shuffle code defined by $C$ and
  \greedy applied to $G-C$ has length at most $m + \lceil T_G^2[d]
  \rceil$ (resp. $m + \lceil T_G^1[d] \rceil$).
\end{proof}
In the following, we show how to compute for an \rtg $G$ a table $T_G[\cdot]$
with
\[
	T_G[d] = \min_{\substack{C \in \mathcal C(G)\\\diff(G - C) = d}} \cost(G - C)
\]
for an arbitrary cost function $\cost(G - C) = c(\sig(G - C))$, where
$c$ is a linear function.  This is done in several steps depending on
whether $G$ is disconnected, is a tree, or is connected and contains a
cycle.  Before we continue, we introduce several preliminaries to
simplify the following calculations.  We denote by $P_s$ a directed
path on $s$ vertices.

\begin{definition}
	A map $f$ that assigns a value to an \odrtg is \emph{signature-linear}
	if there exists a linear function $g \colon \mathbb R^3 \to \mathbb R$
	such that $f(G) = g(\sig(G))$ for every \odrtg $G$.  For a
	signature-linear function $f$,
	$\Delta_f(s) = f(P_{s + 1}) - f(P_s)$ is the \emph{correction
	term}.
\end{definition}
Note that both $\cost = c \circ \sig$ and $\diff = d \circ \sig$ with
$d(X,a_2,a_3) = a_2-a_3$ are signature-linear.  The correction term
$\Delta_f(s)$ describes the change of $f$ when the size of one
connected component is increased from $s$ to $s + 1$.

\begin{lemma}
	\label{lem:signature-linear}
	Let $f$ be a signature-linear function.  Then the following
        hold:
	\begin{compactenum}[(i)]
		\item \label{eqn:sig-union}
			$f(G_1 \cup G_2) = f(G_1) + f(G_2)$ for disjoint \odrtgs $G_1, G_2$,
		\item \label{eqn:sig-new-vertex}
			Let $G = (V, E)$ be an \odrtg and let $v \in V$ with
			in-degree $0$.  Denote by $s$ the size of the connected component
			containing $v$  and let
			$G^+ = (V \cup \{ u \}, E \cup \{ (u, v) \})$ where $u$ is a new
			vertex.  Then $f(G^+) = f(G) + \Delta_f(s)$.
	\end{compactenum}
\end{lemma}

\begin{proof}
  For (\ref{eqn:sig-union}) observe that $\sig(G_1 \cup G_2) =
  \sig(G_1) + \sig(G_2)$; then the statement follows from the
  signature-linearity of $f$.

  For (\ref{eqn:sig-new-vertex}) observe that by adding $u$, we
  replace a connected component of size $s$ by one of size $s+1$.
  Thus $\sig(G^+) = \sig(G) - \sig(P_s) + \sig(P_{s+1})$.  The
  statement follows from the signature-linearity of $f$ and the
  definition of $\Delta_f(s)$.
\end{proof}
Note that $\Delta_f(s) = \Delta_f(s+4)$ for all values of
$s$ and hence it suffices to know
the size of the enlarged component modulo $4$.

The main idea for computing table $T_G[\cdot]$ by dynamic programming
is to decompose $G$ into smaller edge-disjoint subgraphs $G = G_1 \cup
\dots \cup G_k$ such that the copy sets of $G$ can be constructed from
copy sets for each of the $G_i$.  We call such a decomposition
\emph{proper partition} if for every vertex $v$ of $G$ there exists an
index $i$ such that $G_i$ contains all outgoing edges of $v$.  Let
$G_1, \ldots, G_k$ be a proper partition of $G$ and let $\mathcal C_i
\subseteq \mathcal C(G_i)$ for $i = 1, \ldots, k$.  We define
$\mathcal C_1 \otimes \cdots \otimes \mathcal C_k = \left\{ C_1 \cup
  \cdots \cup C_k \mid C_i \in \mathcal C_i, i = 1, \ldots, k
\right\}$.  It is not hard to see that $\mathcal C(G_1 \cup \cdots
\cup G_k) = \mathcal C(G_1) \otimes \cdots \otimes \mathcal C(G_k)$.

\subsection{Disconnected \rtgs}
\label{sec:disconnected-rtgs}

We start with the case that $G$ is disconnected and consists of
connected components $G_1, \ldots, G_k$, which form a proper
partition of $G$.  The main issue is to keep track of $\diff$ and
$\cost$.  For an \rtg $G$, we define $\mathcal C(G; d) = \{ C \in
\mathcal C(G) \mid \diff(G - C) = d \}$.  By
Lemma~\ref{lem:signature-linear}(\ref{eqn:sig-union}) and the
signature-linearity of $\diff$, if $C_i \in \mathcal C(G_i;d_i)$ for
$i=1,2$, then $C_1 \cup C_2 \in \mathcal C(G_1 \cup G_2;d_1+d_2)$.
This leads to the following lemma.

\begin{lemma}
	\label{lem:copyset-calculation-rules}
	Let $G$ be an \rtg and let $G_1, G_2$ be vertex-disjoint \rtgs.  Then
	\begin{compactenum}[(i)]
		\item \label{eqn:CGd} $\mathcal C(G) = \bigcup_d \mathcal C(G; d)$ and
		\item \label{eqn:CGd-union}
			$\mathcal C(G_1 \cup G_2; d) = \bigcup_{d'} \left( \mathcal C(G_1; d')
			\otimes \mathcal C(G_2; d - d') \right)$.
	\end{compactenum}
\end{lemma}

\begin{proof}
  Equation (\ref{eqn:CGd}) follows immediately from the definition of
  $\mathcal C(G;d)$.  For Equation~(\ref{eqn:CGd-union}) observe that
  if $C_1 \in \mathcal C(G_1;d')$ and $C_2 \in \mathcal C(G_2;d-d')$,
  then $C = C_1 \cup C_2$ is a copy set of $G$ and by
  Lemma~\ref{lem:signature-linear}(\ref{eqn:sig-union}) $\diff(G-C) =
  \diff((G_1-C_1) \cup (G_2 -C_2)) = \diff(G_1 - C_1) + \diff(G_2
  -C_2) = d' + d-d' = d$, and hence $C_1 \cup C_2 \in \mathcal
  C(G;d)$.  Conversely, if $C \in \mathcal C(G;d)$, define $C_i = C
  \cap E_i$ where $E_i$ is the edge set of $G_i$ for $i=1,2$.  Let $d'
  = \diff(G_1-C_1)$.  As above, it follows from
  Lemma~\ref{lem:signature-linear}(\ref{eqn:sig-union}) that $d =
  \diff(G-C) = \diff(G_1-C_1) + \diff(G_2-C_2) = d' + \diff(G-C)$, and
  hence $\diff(G-C) = d-d'$.  Thus $C \in \mathcal C(G_1;d') \otimes
  \mathcal C(G_2;d-d')$.
\end{proof}

By further exploiting the signature-linearity of $\cost$, we also get
$\cost((G_1 \cup G_2) - (C_1 \cup C_2)) = \cost(G_1-C_1) +
\cost(G_2-C_2)$, allowing us to compute the cost of copy sets formed
by the union of copy sets of vertex-disjoint graphs.

\begin{lemma}
	\label{lem:components}
	Let $G_1, G_2$ be two vertex-disjoint \rtgs and let $G = G_1 \cup G_2.$
	Then
	$
		T_G[d] = \min_{d'} \{ T_{G_1}[d'] + T_{G_2}[d - d'] \}.
	$
\end{lemma}

\renewcommand{\endproof}{\doendproof}
\begin{proof}
	Applying the definition of $T_G[\cdot]$ as well as
	\Cref{lem:copyset-calculation-rules}~(\ref{eqn:CGd-union}) and
	\Cref{lem:signature-linear}~(\ref{eqn:sig-union}) yields
\begin{eqnarray*}
	T_G[d] & = & \min_{C \in \mathcal C(G; d)} \cost(G - C)
	         =   \min_{C \in \bigcup_{d'} \left( \mathcal C(G_1; d') \otimes \mathcal C(G_2; d - d') \right)}
	             \cost(G - C) \\
	       & = & \min_{d'} \left\{
	             \min_{C \in \mathcal C(G_1; d') \otimes \mathcal C(G_2; d - d')}
	             \cost(G - C) \right\} \\
	       & = & \min_{d'} \left\{
	             \min_{C_1 \in \mathcal C(G_1; d')}
	             \cost(G_1 - C_1)
	             +
	             \min_{C_2 \in \mathcal C(G_2; d - d')}
	             \cost(G_2 - C_2) \right\} \\
	       & = & \min_{d'} \{ T_{G_1}[d'] + T_{G_2}[d - d'] \}. \hspace{6.48cm} \qed
\end{eqnarray*}
\end{proof}
\renewcommand{\endproof}{\hfill\qed\doendproof}

By iteratively applying \Cref{lem:components}, we compute $T_G[\cdot]$
for a disconnected \rtg $G$ with an arbitrary number of connected components.
In the following, we will analyze the running time needed for the combination
of all tables $T_{G_i}[\cdot]$ for the components $G_i$ of $G$.

\begin{lemma}
	\label{lem:running-time-disconnected}
	Let $G$ be an \rtg with $n$ vertices and connected components
	$G_1, \ldots, G_k$.  Given the tables $T_{G_i}[\cdot]$ for $i = 1, \ldots, k$,
	the table $T_G[\cdot]$ can be computed in $O(n^2)$ time.
\end{lemma}

\begin{proof}
  Let $n_i$ denote the number of vertices of $G_i$.  For two graphs
  $H_1$ and $H_2$ with $h_1$ and $h_2$ vertices, respectively,
  computing $T_{H_1 \cup H_2}[\cdot]$ according to
  \Cref{lem:components} takes time $O(h_1 \cdot h_2)$ and the table
  size is $O(h_1 + h_2)$.  Thus, iteratively combining the table for
  $G_{i + 1}$ with the table for $\bigcup_{j = 0}^i G_j$ takes time
  $O(\sum_{i = 1}^{k - 1} n_{i + 1} \sum_{j = 1}^i n_j)$.  It is
  $\sum_{i = 1}^{k - 1} n_{i + 1} \sum_{j = 1}^i n_j \le \sum_{i =
    1}^{k - 1} n_{i + 1} n = n \sum_{i = 1}^{k - 1} n_{i + 1} \le
  n^2$.  Hence, the running time is $O(n^2)$.
\end{proof}

\subsection{Tree \rtgs}
\label{sec:tree-rtgs}

For a tree \rtg $G$, we compute $T_G[\cdot]$ in a bottom-up fashion.
The direction of the edges naturally defines a unique root vertex $r$
that has no incoming edges and we consider $G$ as a rooted tree.  For
a vertex $v$, we denote by $G(v)$ the subtree of $G$ with root $v$.
Let $v$ be a vertex with children $v_1,\dots,v_k$.

How does a copy
set $C$ of $G(v)$ look like?  Clearly, $G(v)-C$ contains precisely one
of the outgoing edges of $v$, say $(v,v_j)$.  Then $Z_j = \{ (v,v_i)
\mid i \ne j\} \subseteq C$.  Graph $G(v)-Z_j$ has connected
components $G(v_i)$ for $i \ne j$, whose union we denote $G_{\neg j}$,
and one additional connected component $G^+(v_j)$ that is obtained
from $G(v_j)$ by adding the vertex $v$ and the edge $(v,v_j)$.  This
forms a proper partition of $G(v) - Z_j$.  As above, we decompose the
copy set $C-Z_j$ further into a union of a copy set $C_{\neg j}$ of
$G_{\neg j}$ and a copy set $C_j$ of $G^+(v_j)$.  Graph $G_{\neg j}$
is disconnected and can be handled as above.  Note that the only child
of the root of $G^+(v_j)$ is $v_j$ and hence $C_j$ is a copy set of
$G(v_j)$.

For expressing the cost
and difference measures for copy sets of $G^+(v_j)$ in terms of copy
sets of $G(v_j)$, we use the correction terms $\Delta_{\cost}$ and
$\Delta_{\diff}$.  By \Cref{lem:signature-linear}~(\ref{eqn:sig-new-vertex}),
$\diff(G^+(v_j)-C_j) = \diff(G(v_j) - C_j)
+ \Delta_{\diff}(s)$, where $s$ is the size of the \emph{root path}
$P(v_j, C_j)$ of
$G(v_j)-C_j$, i.e., the size of the connected component of $G(v_j)-C_j$
containing $v_j$.  An analogous statement holds for $\cost$.
More precisely, it suffices to know $s$ modulo $4$.
Therefore, we further decompose our copy sets as follows, which
allows us to formalize our discussion.

\begin{definition}
  For a tree \rtg $G$ with root $v$ and children $v_1, \ldots, v_k$,
  we define\\
  $\mathcal C(G; d, s) = \{ C \in \mathcal C(G; d) \mid
  |P(v, C)| \equiv s~(\mathrm{mod~} 4) \}$.  We further decompose
  these by $\mathcal C(G; d, s, j) = \{ C \in \mathcal C(G; d, s) \mid
  (v, v_j) \not\in C \}$, according to which outgoing edge of the root
  is not in the copy set.
\end{definition}

\noindent{}The following lemma gives calculation rules for composing copy sets.

\begin{lemma}
	\label{lem:tree-copyset-rules}
	Let $G$ be a tree \rtg with root $v$ and children $v_1, \ldots, v_k$ and
	for a fixed vertex $v_j$, $1 \le j \le k$, let $G^+(v_j)$ be the subgraph
	of $G$ induced by the vertices in $G(v_j)$ together with $v$.
	Let further $G_{\neg j} = \bigcup_{i = 1,i \ne j}^k G(v_i)$ and
	$Z_j = \{ (v, v_i) \mid i \ne j \}$.  Then
	\begin{compactenum}[(i)]
		\item \label{eqn:CGds} $\mathcal C(G; d)    = \bigcup_{s = 0}^3 \mathcal C(G; d, s)$ and
		\label{eqn:CGdsj} $\mathcal C(G; d, s) = \bigcup_{j = 1}^k \mathcal C(G; d, s, j)$.
		\item \label{eqn:CGds-one-child}
			$\mathcal C(G^+(v_j); d, s) = \mathcal C(G(v_j); d - \Delta_{\diff}(s), s - 1)$.
		\item \label{eqn:CGdsj-multiple-children}
			$\mathcal C(G; d, s, j) = \bigcup_{d'} \left( \mathcal C(G_{\neg j}; d') \otimes \mathcal C(G^+(v_j); d - d', s) \otimes \{ Z_j \} \right)$.
	\end{compactenum}
\end{lemma}

\begin{proof}
  The statements in (\ref{eqn:CGds}) follow immediately from the definitions of
  $\mathcal C(G;d,s)$ and $\mathcal C(G;d,s,j)$.
  We continue with Statement~(\ref{eqn:CGds-one-child}).  Since $v$ in $G^+(v_j)$
  has only one child $v_j$, the edge $(v, v_j)$ is not in any copy set
  of $G^+(v_j)$.  Therefore, the copy sets of $\mathcal C(G^+(v_j))$ and $\mathcal
  C(G(v_j))$ are in one-to-one correspondence.  We need to understand
  how the partition into copy sets with difference measure $d$ and
  root path length $s$ (modulo $4$) respects this bijection.  Let $s$ be
  the root path size of $G^+(v_j)-C$ for a copy set $C \in \mathcal C(G^+(v_j))$.
  Obviously, $|P(G(v_j) - C)| = |P(G^+(v_j)-C)| - 1 = s-1$.  Moreover, going
  from $G^+(v_j)-C$ to $G(v_j)-C$ replaces a connected component of size $s$
  by one of size $s-1$.  Therefore $\sig(G(v_j) - C) = \sig(G^+(v_j)-C) -
  \sig(P_s) + \sig(P_{s+1})$.  By the signature-linearity of $\diff$, we
  have $\diff(G(v_j) - C) = \diff(G^+(v_j)-C) - \Delta_{\diff}(s)$.  Note further
  that $\Delta_{\diff}(s) = \Delta_{\diff}(s+4)$ for every value of $s$, and hence
  it suffices to know $s \mod 4$.  Overall, it follows that a copy set
  $C \in \mathcal C(G^+(v_j);d,s)$ is a copy set of $G(v_j)$ with difference
  measure $\diff(G^+(v_j)-C) - \Delta_{\diff}(s)$ and root path size modulo 4 being
  $s-1$.  Thus $C \in C(G(v_j), d - \Delta_{\diff}(s), s-1)$.  And
  conversely $C \in C(G(v_j) , d-\Delta_{\diff}(s),s-1)$ satisfies $C
  \in \mathcal C(G^+(v_j);d,s)$.

  Next, we consider Statement~(\ref{eqn:CGdsj-multiple-children}).
  First, observe that the copy sets~$\mathcal C$ of $G$ whose root path
  starts with $(v,v_j)$ are exactly those copy sets of $G$ that
  contain all edges in $Z_j$.  These sets correspond bijectively to
  copy sets of $G-Z_j$.  Thus $\mathcal C(G;d,s,j) = \mathcal
  C(G-Z_j;d,s) \otimes \{Z_j\}$.  Observe that $G-Z_j = G_{\neg j}
  \cup G^+(v_j)$ is a proper partition of $G - Z_j$.  Furthermore, the
  root path of any copy set of this graph lies in $G^+(v_j)$.
  Therefore, \Cref{lem:copyset-calculation-rules}~(\ref{eqn:CGd-union})
  implies that $\mathcal
  C(G-Z_j;d,s) = \bigcup_{d'} (\mathcal C(G_{\neg j};d') \otimes
  (\mathcal C(G(v_j)^+;d-d',s)$.  Combining this with the previously
  derived description of $\mathcal C(G;d,s,j)$ yields
  Statement~(\ref{eqn:CGdsj-multiple-children}).
\end{proof}

To make use of this decomposition of copy sets,
we extend our table $T$ with an additional parameter $s$ to keep
track of the size of the root path modulo 4.  We call the resulting table
$\tT$.  More formally,
$
	\tT_v[d, s] = \min_{C \in \mathcal C(G(v); d, s)} \cost(G(v) - C).
$
It is not hard to see that $T_G[\cdot]$ can be computed
from $\tT_r[\cdot, \cdot]$ for the root $r$ of a tree \rtg $G$.

\begin{lemma}
	\label{lem:forget-tilde}
	Let $G$ be a tree \rtg with root $r$.  Then $T_G[d] = \min_s \tT_r[d, s]$.
\end{lemma}

\begin{proof}
	Using the definitions of $T_G[\cdot]$ and $\tT_r[\cdot, \cdot]$, we obtain
	\begin{equation*}
		T_G[d] = \min_{C \in \mathcal C(G; d)} \cost(G - C)
		       = \min_{s \in \{ 0, \dots, 3 \}}
		         \min_{C \in \mathcal C(G; d, s)} \cost(G - C)
		       = \min_{s \in \{ 0, \dots, 3 \}} \tT_r[d, s].
	\end{equation*}
\end{proof}

To compute $\tT_v[\cdot, \cdot]$ in a bottom-up fashion, we exploit
the decompositions from \Cref{lem:tree-copyset-rules} and the fact
that we can update the cost function from $G(v_j) - C_j$ to
$G^+(v_j) - C_j$ using the correction term $\Delta_{\cost}$.
The proof is similar to that of \Cref{lem:components}.

\begin{lemma}
	\label{lem:combine-children}
	Let $G$ be a tree \rtg, let $v$ be a vertex of $G$ with children
	$v_1, \ldots, v_k$, and let $G(v_i) = (V_i, E_i)$ for $i = 1, \ldots, k$.
	Let further $G_{\neg j} = (V_{\neg j}, E_{\neg j}) = \bigcup_{i = 1, i \ne j}^k G(v_i)$.
	Then the following equation holds.\\
	$\displaystyle
	\tT_v[d, s] = \min_{j \in \{1, \ldots, k\}} \min_{d'} T_{G_{\neg j}}[d'] + \tT_{v_j}[d - d' - \Delta_{\diff}(s), (s - 1)~\mathrm{mod~} 4] + \Delta_{\cost}(s)$
\end{lemma}

\begin{proof}
  According to the definition of $\tT_v[d, s]$ and
  \Cref{lem:tree-copyset-rules}~(\ref{eqn:CGdsj}), we find that
  \begin{equation}
    \label{eqn:combine-children-start}
    \tT_v[d, s] = \min_{C \in \mathcal C(G; d, s)} \cost(G - C)
    =   \min_j \min_{C \in \mathcal C(G; d, s, j)} \cost(G - C)
  \end{equation}
  Using \Cref{lem:tree-copyset-rules}~(\ref{eqn:CGdsj-multiple-children}) yields
  \begin{equation}
    \min_{C \in \mathcal C(G; d, s, j)} \cost(G - C) =
    \min_{d'} \min_{\substack{X \in \mathcal C(G_{\neg j}; d')\\Y \in \mathcal C(G^+(v_j); d - d', s)}} \cost(G - X - Y - Z_j).
  \end{equation}
  Note that $G - Z_j = G_{\neg j} \cup G^+(v_j)$.  By
  \Cref{lem:copyset-calculation-rules}, we have that for $X \in
  \mathcal C(G_{\neg j}; d'), Y \in \mathcal C(G^+(v_j); d - d', s)$,
  it is $\cost(G - X - Y - Z_j) = \cost(G_{\neg j} \cup G^+(v_j) - X -
  Y) = \cost(G_{\neg j} - X) + \cost(G^+(v_j) - Y)$.  Therefore,
  \begin{equation}
    \begin{array}{>{\displaystyle}l}
      \min_{\substack{X \in \mathcal C(G_{\neg j}; d')\\Y \in \mathcal C(G^+(v_j); d - d', s)}} \cost(G - X - Y - Z_j) \\
      = \min_{X \in \mathcal C(G_{\neg j}; d')} \cost(G_{\neg j} - X) + \min_{Y \in \mathcal C(G^+(v_j); d - d', s)} \cost(G^+(v_j) - Y).
    \end{array}
  \end{equation}
  By definition $\min_{X \in \mathcal C(G_{\neg j}; d')} \cost(G_{\neg
    j} - X) = T_{G_{\neg j}}[d']$.  Furthermore, $G^+(v_j)$ is a tree
  \rtg whose root $v$ has the single child $v_j$.  Hence, by
  \Cref{lem:tree-copyset-rules}~(\ref{eqn:CGds-one-child}) and
  \Cref{lem:signature-linear}~(\ref{eqn:sig-new-vertex}), we find
  \begin{equation}
    \label{eqn:combine-children-end}
    \begin{array}{>{\displaystyle}l}
      \phantom{=}\min_{Y \in \mathcal C(G^+(v_j); d - d', s)} \cost(G^+(v_j) - Y) \\
      = \min_{Y \in \mathcal C(G(v_j); d - d' - \Delta_{\diff}(s), s - 1)} \cost(G(v_j) - Y) + \Delta_{\cost}(s) \\
      = \tT_{v_j}[d - d' - \Delta_{\diff}(s), s - 1] + \Delta_{\cost}(s)
    \end{array}
  \end{equation}
  Combining Equations
  \ref{eqn:combine-children-start}--\ref{eqn:combine-children-end}
  yields the claim.
\end{proof}

For leaves $v$ of a tree \rtg $G$, $\tT_v[0, 1] = 0$ and all other entries
are $\infty$.  We compute $T_G[\cdot]$ by iteratively applying
\Cref{lem:combine-children} in a bottom-up fashion, using
\Cref{lem:forget-tilde} to compute $T[\cdot]$ from $\tT[\cdot, \cdot]$
in linear time when needed.

\begin{lemma}
	\label{lem:running-time-tree}
	Let $G = (V, E)$ be a tree \rtg with $n$ vertices and root $r$.
	The tables $\tT_r[\cdot, \cdot]$ and $T_G[\cdot]$ can be computed in
	$O(n^3)$ time.
\end{lemma}

\begin{proof}
	First observe that given $\tT_v[\cdot, \cdot]$ for $v \in V$, table
	$T_{G(v)}[\cdot]$ can be computed in linear time according to
	\Cref{lem:forget-tilde}.  In particular, $T_G[\cdot]$ can be computed from
	$\tT_r[\cdot, \cdot]$ in linear time.

	We now bound the computation time for $\tT_r[\cdot, \cdot]$.
	Let $v \in V$ with children $v_1, \ldots, v_k$.  Given the tables
	$\tT_{v_i}[\cdot, \cdot]$, we can compute $\tT_v[\cdot, \cdot]$ by
	\Cref{lem:combine-children}.  More precisely, for each $j = 1, \ldots, k$,
	we first compute $T_{G_{\neg j}}[\cdot]$ in quadratic time by
	\Cref{lem:running-time-disconnected} followed by $O(n)$ table lookups,
	one for each value of $d'$.  Hence, processing $v$ takes time
	$O(\deg^+(v) \cdot n^2)$.  Since $\sum_{v \in V} \deg^+(v) = n - 1$,
	the total processing time to compute $\tT_r[\cdot, \cdot]$ in a bottom-up
	fashion is $O(n^3)$.
\end{proof}

\subsection{Connected \rtgs Containing a Cycle}
\label{sec:cycle-rtgs}

We now look at connected \rtgs that contain a cycle.  We first
introduce an additional decomposition for copy sets to simplify the
following calculations.

\begin{lemma}
	\label{lem:cycle-copyset-rules}
	Let $G = (V, E)$ be a connected \rtg containing a directed cycle $K$ and
	let $e_1, \ldots, e_k$ denote the edges of $K$ whose source has out-degree
	at least $2$.  Let further $O = \{ (u, v) \in E \mid u \in K,
	(u, v) \not\in K \}$.
	Then
        \begin{equation*}
          \mathcal C(G;d) = \mathcal C(G - O;d) \otimes \{ O \} \cup \bigcup_{i=1}^k \mathcal C(G - e_i; d) \otimes \{ \{ e_i \} \}.
        \end{equation*}
\end{lemma}

\begin{proof}
  Every copy set $C \in
  \mathcal C(G;d)$ contains either some edge of $K$ or it contains all edges
  in $O$.  Note that edges of $K$ that are not among $e_1, \ldots, e_k$ are not
  contained in any copy set.  Thus, in the former case, $e_i \in C$ for some
  $i \in \{ 1, \ldots, k \}$ and hence
  $C \in \mathcal C(G-e_i;d) \otimes \{ \{ e_i \} \}$.
  In the latter case $C \setminus O$ is a copy set of $G - O$, hence
  $C \in \mathcal C(G - O; d) \otimes \{ O \}$.
  Conversely, any copy set in $\mathcal C(G - O; d) \otimes \{O\}$
  forms a copy set of $G$ and also every copy set in
  $\mathcal C(G-e_i;d) \otimes \{ \{ e_i \} \}$ for any value of $i$ forms a
  copy set of $G$.  This finishes the proof.
\end{proof}

As before, this decomposition can be used to efficiently compute
$T_G[\cdot]$ from the tables of smaller subgraphs of a connected \rtg $G$
containing a cycle.

\begin{lemma}
	\label{lem:cycle-rtg}
	Let $G = (V, E)$ be a connected \rtg containing a directed cycle $K$ and
	let $e_1, \ldots, e_k$ denote the edges of $K$ whose source has out-degree
	at least $2$.  Let further $O = \{ (u, v) \in E \mid u \in K,
	(u, v) \not\in K \}$.
	Then
	\[
		T_G[d] = \min \left\{ T_{G - O}[d],
		\min_{i=1}^k T_{G - e_i}[d] \right\}.
	\]
\end{lemma}

\begin{proof}
	Using the definition of $T_G[\cdot]$ and
	\Cref{lem:cycle-copyset-rules}, we find that
	\begin{eqnarray*}
		T_G[d] & = & \min_{C \in \mathcal C(G; d)} \cost(G - C) \\
		       & = & \min_{C \in (\mathcal C(G - O; d)\otimes\{ O \}) \cup \bigcup_{i=1}^k (\mathcal C(G - e_i; d) \otimes \{ \{ e_i \} \}) } \cost(G - C).
	\end{eqnarray*}
	As we minimize $\cost$ over a union of sets, we can minimize it over
	the sets individually and then take the minimum of the results.
	Hence, we find that
	\begin{equation*}
		\min_{C \in \mathcal C(G - O; d) \otimes \{ O \} } \cost(G - C)
		= \min_{C \in \mathcal C(G - O; d)} \cost(G - O - C)
		= T_{G - O}[d]
	\end{equation*}
	and
	\begin{equation*}
		\min_{C \in \mathcal C(G - e_i; d) \otimes \{ \{ e_i \} \} } \cost(G - C) =
		\min_{C \in \mathcal C(G - e_i; d) } \cost(G - e_i - C) =
		T_{G - e_i}[d],
	\end{equation*}
	which together yield the claim.
\end{proof}

\begin{lemma}
	\label{lem:running-time-cycle}
	Let $G = (V, E)$ be a connected \rtg containing a directed cycle.  The
	table $T_G[\cdot]$ can be computed in $O(n^4)$ time.
\end{lemma}

\begin{proof}
	Let $e_1, \ldots, e_k$ be the edges of the cycle $K$.  First, observe that
	$G - e_i$ is a tree for $i = 1, \ldots, k$.  Hence, we can compute
	each table $T_{G - e_i}[\cdot]$ in $O(n^3)$ time
	by \Cref{lem:running-time-tree}.  Thus, computing all these tables
	takes $O(n^4)$ time.

	Second, let
	$O = \{ (u, v) \in E \mid u \in K, (u, v) \not\in K \}$.  The graph
	$G - O$ is the disjoint union of the cycle $K$ and several tree \rtgs
	$G_1, \ldots, G_t$.  The table $T_K[\cdot]$ has only one finite entry and
	can be computed in constant time.  The tables $T_{G_i}[\cdot]$ can be
	computed in $O(n^3)$ time.  Using \Cref{lem:running-time-disconnected},
	we then compute $T_{G - O}[\cdot]$ in quadratic time.

	With these tables available, we can compute $T_G[\cdot]$ according to
	\Cref{lem:cycle-rtg}.  This takes $O(n^2)$ time.  The overall running time
	is thus $O(n^4)$.
\end{proof}

\subsection{Putting Things Together}
\label{sec:putting-things-together}

To compute $T_G[\cdot]$ for an arbitrary \rtg $G$, we first compute $T_K[\cdot]$
for each connected component $K$ of $G$ using Lemmas~\ref{lem:running-time-tree}
and~\ref{lem:running-time-cycle}.  Then, we compute $T_G[\cdot]$ using
\Cref{lem:running-time-disconnected} and the length of an optimal shuffle code
using \Cref{lem:graph-optimal}.  To actually compute the shuffle code,
we augment
the dynamic program computing $T_G[\cdot]$ such that an optimal copy set $C$
can be found by backtracking in the tables.  An optimal shuffle code is then
constructed by applying \greedy to $G - C$ and adding one copy operation for each
edge in $C$.

\begin{theorem}
	\label{thm:optimal-shuffle-code}
	Given an \rtg $G$, an optimal shuffle code can be computed in $O(n^4)$
	time.
\end{theorem}

\begin{proof}
	We compute all tables $T_C[\cdot]$, where $C$ is a connected component of $G$,
	in $O(n^4)$ time using \Cref{lem:running-time-tree} and
	\Cref{lem:running-time-cycle}.  Using \Cref{lem:running-time-disconnected},
	we then compute $T_G[\cdot]$ in $O(n^2)$ time.  From this, we can compute
	the length of an optimal shuffle code by \Cref{lem:graph-optimal}.

	In fact, it is not difficult to modify the dynamic program in a way that,
	given an entry $T_G[d]$, a corresponding copy set $C$ of $G$ with
	$\cost(G - C) = T_G[d]$ can be computed by backtracking in the tables.
	Hence, to compute an optimal shuffle code for $G$, we first compute an
	optimal copy set $C_{\mathrm{opt}}$ of $G$ in $O(n^4)$ time.  Then, we
	compute an optimal shuffle code
	$\pi_1, \ldots, \pi_k$ for $G - C_{\mathrm{opt}}$ using $\greedy$, which
	takes linear time according to \Cref{thm:greedy-correct}.  Let
	$\pi = \pi_k \circ \ldots \circ \pi_1$.  For each edge
	$(u, v) \in C_{\mathrm{opt}}$, we define a corresponding copy operation
	$\pi(u) \to v$.  Let $c_1, \ldots, c_t$ be these copy operations in
	arbitrary order.  Then the sequence
	$S = \pi_1, \ldots, \pi_k, c_1, \ldots, c_t$ is an optimal shuffle code.
	This can be seen as follows.  First, by \Cref{lem:graph-optimal}, the
	length of $S$ is minimal.  It remains to show that $S$ is indeed a shuffle
	code for $G$.  This is clearly true, as it first shuffles the values in
	the registers so that a subset of the values is in the correct position
	and then uses copy operations to transfer the remaining values to their
	destinations.
\end{proof}

\section{Conclusion}

We have presented an efficient algorithm for generating
optimal shuffle code using copy instructions and permutation instructions,
which allow to arbitrarily permute the contents of up to five registers.
As an intermediate result, we have proven the optimality of the greedy
algorithm for factoring a permutation into a minimal product of
permutations, each of which permutes up to five elements.  It would
be interesting to allow permutations of larger size.

\paragraph{Acknowledgments.}
This work was partly supported by the German Research Foundation (DFG)
as part of the Transregional Collaborative Research Center
\enquote{Invasive Computing} (SFB/TR 89).

\bibliographystyle{splncs03}
\bibliography{bibliography}

\begin{thebibliography}{1}
\providecommand{\url}[1]{\texttt{#1}}
\providecommand{\urlprefix}{URL }

\bibitem{Blazy2009}
Blazy, S., Robillard, B.: Live-range unsplitting for faster optimal coalescing.
  In: Languages, Compilers, and Tools for Embedded Systems (LCTES '09). pp.
  70--79. ACM (2009)

\bibitem{Bouchez2007}
Bouchez, F., Darte, A., Rastello, F.: On the complexity of register coalescing.
  In: Code Generation and Optimization (CGO '07). pp. 102--114. IEEE (2007)

\bibitem{Caprara97}
Caprara, A.: Sorting by reversals is difficult. In: Computational Molecular
  Biology (RECOMB'97). pp. 75--83. ACM (1997)

\bibitem{Farnoud2012}
Farnoud, F., Milenkovic, O.: Sorting of permutations by cost-constrained
  transpositions. IEEE Transactions on Information Theory  58(1),  3--23 (2012)

\bibitem{Grund2007}
Grund, D., Hack, S.: A fast cutting-plane algorithm for optimal coalescing. In:
  Krishnamurthi, S., Odersky, M. (eds.) Compiler Construction, Lecture Notes in
  Computer Science, vol. 4420, pp. 111--125. Springer Berlin Heidelberg (2007)

\bibitem{HackThesis}
Hack, S.: Register Allocation for Programs in {SSA} Form. Ph.D. thesis,
  Universit{\"{a}}t Karlsruhe (2007),
  \url{http://digbib.ubka.uni-karlsruhe.de/volltexte/documents/6532}

\bibitem{Hack2008}
Hack, S., Goos, G.: Copy coalescing by graph recoloring. SIGPLAN Notices
  43(6),  227--237 (2008)

\bibitem{mohr13cases}
Mohr, M., Grudnitsky, A., Modschiedler, T., Bauer, L., Hack, S., Henkel, J.:
  Hardware acceleration for programs in {SSA} form. In: Compilers, Architecture
  and Synthesis for Embedded Systems (CASES '13). ACM (2013)

\bibitem{seress2003}
Seress, {\'A}.: Permutation Group Algorithms, vol. 152. Cambridge University
  Press (2003)

\end{thebibliography}

\end{document}